\documentclass[12pt,a4paper]{article}
\usepackage[latin1]{inputenc}
\usepackage{amsmath}
\usepackage{amsthm}
\usepackage{amsfonts}
\usepackage{amssymb}
\usepackage{graphicx}
\usepackage{nicefrac}
\usepackage{fullpage}
\usepackage{bm}
\usepackage{thm-restate}
\usepackage{comment}

\usepackage{xcolor}
\usepackage[linesnumbered,ruled]{algorithm2e}
\usepackage{tikz}
\usepackage{mathtools}
\usepackage{hyperref}
\usepackage{cleveref}

\newtheorem{thm}{Theorem}[section]
\newtheorem{prop}[thm]{Proposition}
\newtheorem{lemma}[thm]{Lemma}
\newtheorem{remark}[thm]{Remark}
\newtheorem{cor}[thm]{Corollary}

\newtheorem{claim}[thm]{Claim}

\DeclareMathOperator{\tr}{tr}

\begin{document}
\title{Points-Polynomials Incidence Theorem with an Application to Reed-Solomon Codes}

\author{
     Itzhak Tamo \thanks{
      Itzhak Tamo is with the Department of Electrical Engineering--Systems, Tel Aviv University, Tel Aviv, Israel. e-mail:	 tamo@tauex.tau.ac.il.  
      This work was supported by the European Research Council (ERC grant number 852953).} }
	
\maketitle

\begin{abstract}
This paper focuses on incidences over finite fields, extending  to higher degrees a result by Vinh \cite{VINH20111177} on the number of point-line incidences in the plane $\mathbb{F}^2$, where $\mathbb{F}$ is a finite field. Specifically, we present a bound on the number of incidences between points and polynomials of bounded degree in $\mathbb{F}^2$. Our approach employs a singular value decomposition of the incidence matrix between points and polynomials, coupled with an analysis of the related group algebras. This bound is then applied to coding theory, specifically to the problem of average-radius list decoding of Reed-Solomon (RS) codes. We demonstrate that RS codes of certain lengths are average-radius list-decodable with a constant list size, which is dependent on the code rate and the distance from the Johnson radius. While a constant list size for list-decoding of RS codes in this regime was previously established, its existence for the stronger notion of average-radius list-decoding was not known to exist.

\end{abstract}

\section{Introduction}
Discrete geometry is a dynamic field of research with applications in computer science, particularly in computational geometry \cite{sheffer_2022} and coding theory \cite{Dvir2012IncidenceTheorems}. A key area of focus within this field is the study of incidence theorems, which explore the relationships between geometric objects such as points, lines, and planes, with an emphasis on their intersections and configurations.

Consider a set of points $P$ and a set of lines $L$ in the plane $\mathbb{F}^2$ for a given field $\mathbb{F}$. An incidence is defined as a pair $(p, \ell) \in P \times L$ where the point $p$ lies on the line $\ell$. Incidence theorems aim to provide bounds on the number of such incidences, denoted by $I(P,L)$, dependent on the number of points in $P$ and lines in $L$.

A prominent example of such a theorem is the renowned Szemer\'{e}di-Trotter theorem \cite{Szemeredi1983Extremal,Szemeredi1983Combinatorial}, which establishes bounds on the number of incidences between points and lines in the plane $\mathbb{R}^2$. While this type of theorems have been extensively studied over $\mathbb{R}$ and to some extent over $\mathbb{C}$ \cite{Toth2015}, the scenario in finite fields presents additional challenges, making many incidence problems more complex. This paper focuses on incidences over finite fields.

 Bourgain, Katz, and Tao \cite{BourgainKatzTao} established that the number of incidences between a set of $N$ points and $N$ lines in $\mathbb{F}^2$, for a prime field $\mathbb{F}$ of order $q$ and $N = q^\alpha$ with $\alpha < 2$, is bounded by $O(N^{3/2-\varepsilon})$ for some $\varepsilon$ dependent on $\alpha$. The exact relationship between $\varepsilon$ and $\alpha$, however, remains challenging to ascertain. In cases where $N = O(\log \log \log (q))$, Grosu \cite{Grosu2014} demonstrated that the point and line sets can be embedded into $\mathbb{C}^2$, while preserving many incidences. By applying the Szemer\'{e}di-Trotter theorem in $\mathbb{C}^2$ \cite{Toth2015}, he then derived that $I(P, L) = O(N^{4/3})$. Vinh \cite{VINH20111177}, utilizing spectral graph theory methods proved the following.
\begin{thm}(\cite[Theorem 3]{VINH20111177})
\label{vinh-thm}
    Let $P$ be a collection of points and $L$ be a collection of lines in $\mathbb{F}^2$, where $\mathbb{F}$ is a prime field of order $q$. Then, we have
\[
I(P,L) \leq \frac{|P||L|}{q}+\sqrt{|P||L|q} .
\]
\end{thm}
Specifically, \Cref{vinh-thm} was derived using the expander-mixing lemma and an analysis of the eigenvalues of the corresponding incidence graph.

We adopt a methodology similar to that in \cite{VINH20111177}, extending \Cref{vinh-thm} from the scenario of lines in a plane to the case where these lines are substituted by arbitrary polynomials of bounded degree over any finite field. We establish the following theorem, where \(\mathbb{F}\) denotes any finite field of order \(Q\), \(Q\) being a prime power, and $\mathbb{F}^k[x]$ is the space of polynomials of degree less than $k\in \mathbb{N}$. 
\begin{restatable}{thm}{mainthm}
\label{main-thm}
Let $L\subseteq \mathbb{F}^k[x]$ and $P\subseteq \mathbb{F}^2$ be a set of polynomials of degree less than $k$ and points in the plane, respectively. Then, the number $I(L,P)$ of incidences  between points and polynomials satisfies 
\begin{equation}
\label{eq:main-result}
\Big|I(L,P)-\frac{|L||P|}{Q}\Big|\leq 
\sqrt{|L||P|(Q+|L|(k-1))(1-\frac{1}{Q})}\leq \sqrt{|L||P|(Q+|L| k)} 
\end{equation} 
\end{restatable}
To derive \Cref{main-thm}, we analyze the group algebras \(\mathbb{C}[\mathbb{F}^k[x]]\) and \(\mathbb{C}[\mathbb{F}^2]\), incorporating some basic character theory. 
In the process of establishing \Cref{main-thm}, we present a singular value decomposition (SVD) of the incidence matrix of the points in the plane and all polynomials of some bounded degree. This SVD, potentially of independent interest, facilitates a modest improvement to \Cref{vinh-thm}, as detailed below.
\begin{restatable}{prop}{vinhimprove}
\label{vinh-improvement}
Let $L\subseteq \mathbb{F}^2[x]$ and $P\subseteq \mathbb{F}^2$ be a set of lines and points in the plane, respectively. Then, the number $I(L,P)$ of incidences between points and lines satisfies 
$$\Big|I(P,L)-\frac{|P||L|}{Q} \Big|\leq \sqrt{(|P||L|Q}(1-\frac{1}{Q}).$$
\end{restatable}

As previously noted, incidence theorems have also been applied in the realm of coding theory. In this context, we present another application by employing \Cref{main-thm} to address the problem of average-radius list-decoding of Reed-Solomon (RS) codes, which we discuss next.

\vspace{0.4cm}
\noindent\textbf{Application to Reed-Solomon codes.} An error-correcting code \( C \subseteq \mathbb{F}^n \), where \(\mathbb{F}\) is a finite field, consists of a subset of vectors known as codewords, each of length \( n \). The primary objective in constructing such a code \( C \) is to facilitate the retrieval of an original codeword \( c \in C \) from its corrupted version, while simultaneously maximizing the size of \( C \). The size of \( C \) is quantified by its rate \( R \), defined as \( R := \frac{\log_{|\mathbb{F}|} |C|}{n} \).

In the standard problem of unique decoding, the goal is to accurately reconstruct a codeword \( c \) from any corrupted version that differs from \( c \) in a limited number of positions. However, there are scenarios where it is necessary to reconstruct \( c \) from a version that is significantly corrupted \cite{elias2006errorcorrecting,blinovskii1986bounds,ahlswede1973channelcapacities}. In such cases, if the code \( C \) is not sufficiently small, the task of uniquely recovering \( c \) becomes impossible. This leads to the concept of list-decoding, first introduced by Elias \cite{peter1957listdecoding} and Wozencraft \cite{wozencraft1958listdecoding}, where unique decoding is no longer feasible. Instead, the objective shifts to identifying a small list of potential codewords from \( C \) that includes the original codeword \( c \).

Algebraic codes are central to the study of list-decoding, serving as primary examples of codes that can be efficiently list-decoded. A particularly notable family within this category is Reed-Solomon (RS) codes \cite{reed1960polynomial}, which are constructed from the evaluations of low-degree polynomials. In the 1990s, efficient list-decoding algorithms for these codes were developed. Sudan \cite{sudan1997decoding}, with subsequent improvements by Guruswami and Sudan \cite{Guru-sudan-algo}, demonstrated that RS codes can be efficiently list-decoded up to the Johnson radius, a specific level of corruption. Further research \cite{CassutoBruck2004,Dvir} has shown that the list size is constant in the rate \( R \) of the code and the distance \( \varepsilon > 0 \) of the list-decoding   from the Johnson radius.

Average-radius list-decoding\cite{GuruswamiNarayanan2014,9611262} is another approach to decoding with a small list size, though it has been considered less in the literature compared to list-decoding. In essence, while a good list-decodable code ensures that not too many codewords are contained within any small Hamming ball, average radius list-decoding requires that for any vector \( v \in \mathbb{F}^n \), there are not too many codewords whose average Hamming distance from \( v \) is small. Upon examining the precise definitions of these two notions, it is clear that average-radius list-decoding is a stronger concept than list-decoding, as it implies the latter.


For RS codes, it is known that a random RS code, with high probability, achieves the capacity of average-radius list-decoding with optimal list size \cite{gopigopi}. In this work, we present a result similar to those found in \cite{CassutoBruck2004,Dvir}, demonstrating that any RS code with a length proportional to its field size is average-radius list-decodable, with a list size  at most $O\left(\frac{1}{\varepsilon \sqrt{R}}\right)$, where $R$ is the code rate and $\varepsilon>0$ is the distance from the Johnson radius, as detailed in \Cref{list-decoding-cor1}.

\vspace{0.4cm}
\noindent{\bf Organization.}
In \Cref{preliminaries}, we provide the necessary background in character theory along with some essential notations. We then proceed in \Cref{right-singular-vectors} to analyze the points-polynomials incidence matrix and derive its right-singular vectors. Continuing in \Cref{left-singular-vectors}, we identify the left-singular vectors of the incidence matrix. Finally, in \Cref{applicatoins}, we present our main results,  the bound on the number of incidences between points and polynomials, which is then applied to the problem of average-radius list-decoding of RS codes.

\section{Preliminaries and Notations}
\label{preliminaries}
Let $\mathbb{F}$ be a finite field of size $Q$ and let  $\mathbb{F}^k[x]$ be the space of polynomials of degree less than $k$ over  $\mathbb{F}$, which we identify with the vector space $\mathbb{F}^k$. 
If $\mathbb{F}$ is a degree $m$ extension over $\mathbb{F}_q$ the prime field of size $q$, then the trace function $\tr:\mathbb{F}\rightarrow \mathbb{F}_q$ is defined as $$\tr(x)=x+x^q+\ldots+x^{q^{m-1}}.$$
Due to the linearity of the  Frobenius-homomorphism, it follows that the trace function is also linear, i.e., $\tr(x+y)=\tr(x)+\tr(y).$ Further, it is well-known that the set of characters (group-homomorphisms) from $\mathbb{F}^k$ to $\mathbb{C}^*$ are defined by the vectors of $\mathbb{F}^k$ as follows. A vector  $v
\in \mathbb{F}^k$ defines the character   
$$\chi_v(x):=e(\tr(<v,x>)):\mathbb{F}^k\rightarrow \mathbb{C}^*,$$
where $<x,v>=\sum_i x_iv_i$ is the standard inner product between vectors, and $e(x)=e^{\frac{2\pi i x}{q}}$. 

It is easy to verify that $\chi_v(x)$ indeed defines a group homomorphism, as for any $x,y \in \mathbb{F}$ 
\begin{align*}
\chi_v(x+y)&=e(\tr(<v,x+y>))=e(\tr(<v,x>+<v,y>))\\
&=e(\tr(<v,x>))\cdot e(\tr(<v,y>))=\chi_v(x)\cdot \chi_v(y)
\end{align*}

\vspace{0.4cm}
\noindent\textbf{The Points-Polynomials Incidence Graph.} 
Consider the following incidence bipartite graph, derived from the point-polynomial incidences, as follows. 
The vertex parts are 
 $\mathbb{F}^k[x]$, the set of all polynomials over $\mathbb{F}$ of degree less than $k$,  and $\mathbb{F}^2$, the $Q^2$ points in the plane.

A polynomial (vertex) $f \in \mathbb{F}^k[x]$ is connected to a point $(\alpha, \beta) \in \mathbb{F}^2$ if and only if $f(\alpha) = \beta$. It is straightforward to verify that the degree of a polynomial vertex and a point vertex in the plane is $Q$ and $Q^{k-2}$, respectively.

 Let $T$ be the $Q^k \times Q^2$ points-polynomials incidence matrix (derived from the above bipartite graph) over $\mathbb{C}$. The rows and columns of $T$ are indexed by all polynomials of degree less than $k$ and points in the plane, respectively. The entry $T_{f,(x,y)}$ for a polynomial $f \in \mathbb{F}^k[x]$ and a point $(x,y) \in \mathbb{F}^2$ is $1$ if and only if the graph of $f$ passes through the point $(x,y)$, i.e., $f(x) = y$; otherwise, it is $0$.

Our main result will follow by determining the singular value decomposition (SVD) of the  matrix $T$, including the singular values and their corresponding left and right singular vectors. We then apply this decomposition in conjunction with an analysis similar to the expander mixing lemma \cite{ALON198815}. To this end, we will explore the eigenvalues and eigenvectors of the symmetric matrices $T \cdot T^*$ and $T^* \cdot T$.

\begin{remark}
    One might question why we do not directly apply the expander mixing lemma \cite{ALON198815} to the aforementioned incidence graph. The rationale is that employing it as a black box, akin to the method used in \cite{VINH20111177}, results in a weaker bound than what can be achieved through a more careful analysis of the matrix $T$. This is evident in the modest improvement we achieve over \cite[Theorem 3]{VINH20111177}. However, it is noteworthy that this improvement becomes more pronounced as the degree of the polynomials increases.
\end{remark}

\vspace{0.4cm}
\noindent\textbf{Singular Value Decomposition.} Let $T$ be an $m \times n$ real or complex matrix (assume w.l.o.g that $m \geq n$). In our case, $T$ is an $Q^k\times Q^2$ matrix. The Singular Value Decomposition (SVD) of $T$ is a factorization of the form $T = U\Sigma V^*$, where:
\begin{itemize}
    \item $U$ is an $m \times m$ complex unitary matrix,
    \item $\Sigma$ is an $m \times n$ rectangular diagonal matrix with non-negative real numbers on the diagonal,
    \item $V$ is an $n \times n$ complex unitary matrix.
\end{itemize}
The nonzero elements $\sigma_{i,i}$ on the main diagonal of $\Sigma$ are the singular values of $T$. The columns of $U$ (left-singular vectors) and $V$ (right-singular vectors) form two sets of orthonormal bases $u_1, \ldots, u_m$ and $v_1, \ldots, v_n$, respectively.

Writing $\Sigma = 
\begin{pmatrix}
D\\
0_{(m-n) \times n}
\end{pmatrix}
$, where $D$ is an $n \times n$ diagonal matrix and $0_{(m-n) \times n}$ is an $(m-n) \times n$ zero matrix, we have:
\[ T \cdot T^* = U\Sigma V^*(U\Sigma V^*)^* = U\Sigma \Sigma^* U^* = U \begin{pmatrix} D^2 & 0 \\ 0 & 0 \end{pmatrix} U^*, \]
and
\[ T^* \cdot T = V \Sigma^* \Sigma V^* = V D^2  V^*. \]
It can be easily verified that both $T \cdot T^*$ and $T^* \cdot T$ are symmetric positive semi-definite matrices and have the same positive eigenvalues. Furthermore, a positive  $\lambda $ is an eigenvalue of $T^* \cdot T$ corresponding to an eigenspace of dimension $d$ if and only if it is also an eigenvalue of $T \cdot T^*$ corresponding to an eigenspace of the same dimension $d$. This equivalence holds if and only if the matrix $T$ has exactly $d$ singular values equal to $\sqrt{\lambda}$.
.
\begin{remark}
In spectral graph theory, it is customary to consider the adjacency matrix of a graph and analyze its eigenvalues and eigenvectors. However, in our approach, we find it more convenient to work directly with the points-polynomials incidence matrix $T$ and analyze it via its singular value decomposition. These two matrices are closely related, as can be seen by noting that the matrix 
\[ A = \begin{pmatrix}
0 & T \\
T^* & 0
\end{pmatrix}, \]
is the adjacency matrix of the incidence graph. Furthermore, the relationship between the eigenvalues and eigenvectors of $A$ and the SVD of $T$ is well-established (see, for example, the discussion before Theorem 5.1 in \cite{HAEMERS1995593}) and is summarized in the following claim.

\begin{claim}
Let $A = \begin{pmatrix}
    0 & T \\
    T^* & 0
\end{pmatrix}$ be an $n \times n$ symmetric matrix, with eigenvalues $\lambda_1 \geq \lambda_2 \geq \ldots \geq \lambda_n$. Then, $\lambda_i = -\lambda_{n-i+1}$ for $i = 1, \ldots, n$. Furthermore, the positive eigenvalues of $A$ correspond to the singular values of $T$; they are also the square roots of the nonzero eigenvalues of $T \cdot T^*$ and $T^* \cdot T$.
\end{claim}

Therefore, deriving the SVD of $T$ is equivalent to obtaining the spectral decomposition of $A$.
\end{remark}


\vspace{0.4cm}
\noindent\textbf{Group Algebra.} For the abelian group $\mathbb{F}^k$, its group algebra over $\mathbb{C}$, denoted by $\mathbb{C}[\mathbb{F}^k]$, consists of all possible formal sums of the variables $x_g, g \in \mathbb{F}^k$, over $\mathbb{C}$. That is, 
\[\mathbb{C}[\mathbb{F}^k] = \left\{ \sum_{g \in \mathbb{F}^k} a_g x_g : a_g \in \mathbb{C} \right\}.\] 
Addition in this algebra is defined component-wise. For elements $\sum_{g \in \mathbb{F}^k} a_g x_g$ and  $\sum_{g \in \mathbb{F}^k} b_g x_g$  in  $\mathbb{C}[\mathbb{F}^k]$, we have 
\[\sum_{g \in \mathbb{F}^k} a_g x_g + \sum_{g \in \mathbb{F}^k} b_g x_g := \sum_{g \in \mathbb{F}^k} (a_g + b_g) x_g,\]
while multiplication is defined via the addition in the group $\mathbb{F}^k$:
\[\sum_{g \in \mathbb{F}^k} a_g x_g \cdot \sum_{g \in \mathbb{F}^k} b_g x_g := \sum_{g \in \mathbb{F}^k} \left( \sum_{\tilde{g} \in \mathbb{F}^k} a_{g-\tilde{g}} b_{\tilde{g}} \right) x_g.\]

Each element $f = \sum_{g \in \mathbb{F}^k} a_g x_g \in \mathbb{C}[\mathbb{F}^k]$ can be viewed as a function $f: \mathbb{F}^k \rightarrow \mathbb{C}$ (or as a vector of length $|\mathbb{F}^k|$, indexed by the elements of $\mathbb{F}^k$), by simply setting $f(g) = a_g$ for any $g \in \mathbb{F}^k$. Thus, the group algebra is a vector space of dimension $|\mathbb{F}^k|$ over $\mathbb{C}$, and the elements $x_g,g\in \mathbb{F}^k$ is a basis of $\mathbb{C}[\mathbb{F}^k]$ called the standard basis. Another important basis of the group algebra is its   set of characters $\chi_v,v\in \mathbb{F}^k$, which in fact forms an orthogonal basis. Indeed,  there are $|\mathbb{F}^k|$ characters, one for each $v \in \mathbb{F}^k$, and they are linearly independent as they are orthogonal under the inner product defined as 
\[ \langle f, g \rangle = \sum_{x \in \mathbb{F}^k} f(x) \cdot \overline{g(x)}, \]
for any $f, g \in \mathbb{C}[\mathbb{F}^k]$. While it is customary to define a normalized inner product on $\mathbb{C}[\mathbb{F}^k]$ where each character has norm one, we find it more convenient for our purposes to use the standard inner product.

\vspace{0.4cm}
\noindent\textbf{Multiplication Operators of the Group Algebra.} Given an element $z \in \mathbb{C}[\mathbb{F}^k]$ of the group algebra, one can define the linear operator $\phi_z: \mathbb{C}[\mathbb{F}^k] \rightarrow \mathbb{C}[\mathbb{F}^k]$ by multiplication by $z$, i.e., for any element $f \in \mathbb{C}[\mathbb{F}^k]$,
\[\phi_z(f) := z \cdot f = f \cdot z.\]
The second equality holds since the group $\mathbb{F}^k$ is commutative. The following proposition, a well-known result, shows that the characters of the multiplication operator form a basis of eigenvectors. For completeness, we include its short proof.

\begin{prop}
\label{char-eigenvector}
Let $z \in \mathbb{C}[\mathbb{F}^k]$. Then, any character $\chi \in \mathbb{C}[\mathbb{F}^k]$ is an eigenvector of $\phi_z$ with eigenvalue $\langle z, \chi \rangle$.
\end{prop}

\begin{proof}
The result follows from the following calculation:
\begin{align*}
\phi_z(\chi) &= z \cdot \chi = \left( \sum_{g \in \mathbb{F}^k} z_g x_g \right) \left( \sum_{\tilde{g} \in \mathbb{F}^k} \chi(\tilde{g}) x_{\tilde{g}} \right) \\
&= \sum_{g \in \mathbb{F}^k} \sum_{\tilde{g} \in \mathbb{F}^k} z_{\tilde{g}} \chi(g - \tilde{g}) x_g \\
&= \sum_{g \in \mathbb{F}^k} \sum_{\tilde{g} \in \mathbb{F}^k} z_{\tilde{g}} \overline{\chi(\tilde{g})} \chi(g) x_g \\
&= \sum_{g \in \mathbb{F}^k} \langle z, \chi \rangle \chi(g) x_g = \langle z, \chi \rangle \cdot \chi,
\end{align*}
completing the proof.
\end{proof}

\section{Right Singular Vectors of the Incidence Matrix}
\label{right-singular-vectors}
In this section, we give a complete description of the right singular vectors of $T$ by considering the square matrix of order $Q^2$  
$$A:=T^*T=(V\Sigma^* U^*)(U\Sigma V^*)=VD^2V^*.$$
Note that the rows and columns of $A$ are indexed by the $Q^2$ points of the 
plane.  By the definition of $T$ one can easily verify that for any two points 
$(x,y),(w,z)\in \mathbb{F}^2$
$$A_{(x,y),(w,z)}=\begin{cases}
Q^{k-1} & (x,y)=(w,z)\\
Q^{k-2} & x\neq w \\
0 & x =w \text{ and } y\neq z. 
\end{cases}$$
The following lemma shows that the matrix $A$ can be viewed as a representation of some multiplication operator under the standard basis of $\mathbb{C}[\mathbb{F}^2]$.

\begin{lemma}
Let 
$$f(x)=\sum_{(v_1,v_2)\in \mathbb{F}^2}a_{(v_1,v_2)}x_{(v_1,v_2)}\in \mathbb{C}[\mathbb{F}^2],\text{ where } a_{(v_1,v_2)}=\begin{cases}
Q^{k-1} & (v_1,v_2)=(0,0)\\
Q^{k-2} & v_1\neq 0\\
0 & v_1=0,v_2\neq 0.
\end{cases}$$
Then, $A$ is the representation of $\phi_f$ under the standard basis $B=\{x_v:v\in \mathbb{F}^2\}$, i.e., $[\phi_f]_B^B=A$.
\end{lemma}

\begin{proof}
Let $x_{(\gamma,\delta)}\in B$ be a standard basis vector, then 
by definition 
\begin{align*}
\phi_f(x_{(\gamma,\delta)})&=f(x)\cdot x_{(\gamma,\delta)}=
\sum_{(\alpha,\beta)\in \mathbb{F}^2}a_{(\alpha,\beta)}x_{(\alpha+\gamma,\beta+\delta)}\\
&=
\sum_{(\alpha,\beta)\in \mathbb{F}^2}a_{(\alpha-\gamma,\beta-\delta)}x_{(\alpha,\beta)}
= \sum_{(\alpha,\beta)\in \mathbb{F}^2}A_{(\alpha,\beta),(\gamma,\delta)}x_{(\alpha,\beta)},
\end{align*} 
and the result follows. 
\end{proof}
By  \cref{char-eigenvector}, we conclude that the characters of $\mathbb{F}^2$ form a basis of eigenvectors. We proceed by calculating their  eigenvalue. 
\begin{prop}
\label{prop.3.2}
A character $\chi_{(u_1,u_2)}\in \mathbb{C}[\mathbb{F}^2]$ of $\mathbb{F}^2$ is an eigenvector of $\phi_f$ with eigenvalue 
$$\lambda= \begin{cases}
Q^k & (u_1,u_2)=0,\\
Q^{k-1} & u_2\neq 0,\\
0 & \text{else.}
\end{cases}$$
\end{prop}

\begin{proof}
By  \Cref{char-eigenvector} $\lambda=<f,\chi_{(u_1,u_2)}>$. Then

\begin{align}
\label{good6}
\overline{\lambda}= \overline{<f,\chi_{(u_1,u_2)}>}&=\sum_{(\alpha,\beta)\in \mathbb{F}^2} f(\alpha,\beta)\cdot\chi_{(u_1,u_2)}(\alpha,\beta)\nonumber\\
&=Q^{k-1}+Q^{k-2}\sum_{\substack{(\alpha,\beta)\in \mathbb{F}^2\\ \alpha\neq 0}}e(\tr(\alpha u_1 + \beta u_2))
\end{align}
Clearly, if $u_1,u_2=0$ then \eqref{good6}
 becomes 
 $$Q^{k-1}+Q^{k-2}\sum_{\substack{(\alpha,\beta)\in \mathbb{F}^2\\ \alpha\neq 0}}1=Q^k,$$ as needed. 
 Next, if $u_2\neq 0$ then \eqref{good6} becomes 
$$Q^{k-1}+Q^{k-2}\sum_{\alpha\neq 0 }\sum_{\beta\in \mathbb{F}}e(\tr(\alpha u_1+\beta u_2 ))=Q^{k-1},$$
 as the inner sum vanishes when $\beta$ runs over all the field elements. Lastly, if $u_1\neq 0$ but $u_2=0$, then \eqref{good6} becomes 
 
$$Q^{k-1}+Q^{k-2}\sum_{\beta\in \mathbb{F}}\sum_{\alpha\neq 0}e(\tr(\alpha u_1)=Q^{k-1}+Q^{k-2}\sum_{\beta\in \mathbb{F}}-1=0.$$
\end{proof}
The following theorem  follows easily from \Cref{prop.3.2}.
\begin{thm}
\label{eigenvalues2}
The characters  of $\mathbb{F}^2$ are eigenvectors of the operator $\phi_f$ with the following eigenvalues: 
\begin{center}
\begin{tabular}{ |c |c| c|  }
\hline
 \text{Character $\chi_{(u_1,u_2)}$} & \text{No. of  characters} & \text{Eigenvalue} \\ 

  \hline
 $(u_1,u_2)=0$  & $1$& $Q^{k}$  \\  
\hline  $(u_1,u_2), u_2\neq 0$  & $Q(Q-1)$&$Q^{k-1}$\\
  \hline $(u_1,0), u_1\neq 0$ &  $Q-1$&$0$ \\
   \hline
\end{tabular}
\end{center}
\end{thm}
We immediately get the following corollary. 
\begin{cor}
\label{cor-right-singular-vectors}
The normalized characters of $\mathbb{F}^2$ serve as right-singular vectors of the matrix $T$, each associated with a singular value equal to the square root of their corresponding eigenvalue, as detailed in Theorem \ref{eigenvalues2}.
\end{cor}

Building on \Cref{eigenvalues2}, we next establish a bound on the squared norm of the projection of an indicator vector onto the eigenspaces. This bound, presented in the following lemma, will be instrumental in obtaining the results of  Section \ref{applicatoins}.

\begin{lemma}
\label{decomposition2} 
Let $1_P\in \mathbb{C}^{Q^2}$ be the indicator vector of a set $P\subseteq \mathbb{F}^2$ of $p$ points in the plane. Then, the squared norm of the projection of $1_P$ onto  the eigenspaces with eigenvalues $Q^k$ and $Q^{k-1}$ is $\frac{p^2}{Q^2}$ and at most $\frac{(Q-1)p}{Q}$, respectively.  
\end{lemma}

\begin{proof}
The squared norm of the projection  of $1_P$ onto the normalized trivial character is 
\[ \left|\left\langle 1_P, \frac{\chi_0}{Q} \right\rangle\right|^2 = \frac{p^2}{Q^2}. \]
Next, we bound the squared norm of the projection of $1_P$ onto the eigenspace corresponding to the eigenvalue $Q^{k-1}$, which is calculated as follows:
\begin{align*}
&\sum_{\substack{u_1, u_2 \in \mathbb{F} \\ u_2 \neq 0}} \left| \left\langle 1_P, \frac{\chi_{(u_1, u_2)}}{Q} \right\rangle \right|^2 \\
&= Q^{-2} \sum_{\substack{u_1, u_2 \in \mathbb{F} \\ u_2 \neq 0}} \sum_{(\alpha, \beta) \in P} \chi_{(u_1, u_2)}(\alpha, \beta) \sum_{(\alpha', \beta') \in P} \overline{\chi_{(u_1, u_2)}(\alpha', \beta')} \\
&= Q^{-2} \sum_{\substack{u_1, u_2 \in \mathbb{F} \\ u_2 \neq 0}} \sum_{(\alpha, \beta), (\alpha', \beta') \in P} \chi_{(u_1, u_2)}(\alpha - \alpha', \beta - \beta') \\
&= \frac{(Q-1)p}{Q} + Q^{-2} \sum_{\substack{(\alpha, \beta), (\alpha', \beta') \in P \\ (\alpha, \beta) \neq (\alpha', \beta')}} \sum_{\substack{u_1, u_2 \in \mathbb{F} \\ u_2 \neq 0}} e(\tr(u_1(\alpha - \alpha') + u_2(\beta - \beta'))) \\
&\leq \frac{(Q-1)p}{Q},
\end{align*}
where the inequality follows from the fact that if $\alpha \neq \alpha'$, the inner sum vanishes as $u_1$ ranges over $\mathbb{F}$. Conversely, if $\alpha = \alpha'$, then $\beta$ must necessarily differ from $\beta'$, resulting in the inner sum equalling $-Q$. This is because, for a fixed $u_1$, the sum is $-1$ as $u_2$ varies over $\mathbb{F} \setminus \{0\}$. In both scenarios, the double sum is non-positive, leading to the stated bound.
\end{proof}

\section{Left Singular Vectors of the Incidence Matrix}
\label{left-singular-vectors}
In this section we take a similar approach to Section \ref{right-singular-vectors} and provide a complete description of the left singular vectors of $T$ by considering the square matrix  
$A = T \cdot T^*$ of order $Q^k$. Notice that  the rows and columns of $A$ can be  indexed by polynomials in $\mathbb{F}^k[x]$, and  it can be easily verified that for two polynomials $f, g \in \mathbb{F}^k[x]$, the entry $A_{f,g}$ is the number of points $\alpha \in \mathbb{F}$ where $f$ and $g$ agree, i.e.,
\[ A_{f,g} = |\{\alpha \in \mathbb{F} : f(\alpha) = g(\alpha)\}|. \]
Consequently, the diagonal entries of $A$ are equal to $Q$, that is, $A_{f,f} = Q$ for any $f \in \mathbb{F}^k[x]$.

Since the groups \(\mathbb{F}^k\) and \(\mathbb{F}^k[x]\) are isomorphic, their corresponding group algebras \(\mathbb{C}[\mathbb{F}^k]\) and \(\mathbb{C}[\mathbb{F}^k[x]]\) are also isomorphic. Consequently, in the sequel, we will refer to them interchangeably as the same mathematical object.
 
Consider $z \in \mathbb{C}[\mathbb{F}^k[x]]$, a group algebra element that records the number of zeros each polynomial $g \in \mathbb{F}^k[x]$ has over the field $\mathbb{F}$. Specifically, let $z = \sum_{g \in \mathbb{F}^k[x]} z_g x_g$ where 
\[ z_g = |\{\alpha \in \mathbb{F} : g(\alpha) = 0\}|. \]
The following lemma shows that the matrix $A$ represents the multiplicatoin operator $\phi_z$ under the standard basis of the group algebra.

\begin{lemma}
The matrix $A=T\cdot T^*$ is  a representation of the operator $\phi_z$ under  the standard basis $B = \{x_g : g \in \mathbb{F}^k\}$, i.e.,
\[ [\phi_z]_B^B = A. \]
\end{lemma}

\begin{proof}
Let $g$ be a polynomial of degree less than $k$, then by definition 
$$\phi_z(x_g)=z\cdot x_g=\Big(\sum_{\tilde{g}\in \mathbb{F}^k} z_{\tilde{g}} x_{\tilde{g}}\Big)x_g=
\sum_{\tilde{g}\in \mathbb{F}^k}z_{\tilde{g}}x_{\tilde{g}+g}=\sum_{f\in \mathbb{F}^k}z_{f-g}x_f.$$
Therefore, the coefficient of $x_f$ represents the number of zeros that the polynomial $f-g$ has over $\mathbb{F}$. Equivalently, this is the number of points over $\mathbb{F}$ where the polynomials $f$ and $g$ agree, which corresponds precisely to the $(f,g)$ entry of the matrix $A$. Thus, the result follows.
\end{proof}
%
Having established that the matrix $A$ represents the linear operator $\phi_z$ under the standard basis, we can conclude by \Cref{char-eigenvector} that the characters of $\mathbb{F}^k$ are the eigenvectors of $A$ (where each character is represented as a vector under the standard basis). The following lemma provides a characterization of the characters that correspond to nonzero eigenvalues.

\begin{lemma}
\label{lemma3}
Let $\chi\in \mathbb{C}[\mathbb{F}^k[x]]$ be a nontrivial character of $\mathbb{F}^k[x]$. Then $\chi$ is an eigenvector of $\phi_z$ with  eigenvalue  $Q^{k-1}$ if there exists $\alpha \in \mathbb{F}$ such that for any polynomial $f$ of degree less than $k-1$,
\[ \chi(xf(x)) = \chi(\alpha f(x)). \]
Otherwise, the eigenvalue is zero.
\end{lemma}

\begin{proof}
Assume that there exists such  $\alpha\in \mathbb{F}$, then by Proposition \ref{char-eigenvector} the eigenvalue equals to the  inner product $<z,\chi>$. Then,   
\begin{align*}
\overline{<z,\chi>}&=\sum_{f\in \mathbb{F}^k[x]}z(f)\chi(f)=\sum_{f\in \mathbb{F}^k[x]}\sum_{\substack{\beta \in \mathbb{F}\\ f(\beta)=0}}\chi(f) =\sum_{\beta \in \mathbb{F}}\sum_{\substack{f\in \mathbb{F}^k[x] \\f(\beta)=0}}\chi(f)=\sum_{\beta \in \mathbb{F}}\sum_{f\in \mathbb{F}^{k-1}[x]}\chi((x-\beta)f)\\
&=\sum_{\beta \in \mathbb{F}}\sum_{f\in \mathbb{F}^{k-1}[x]}\chi(xf)\chi(-\beta f)=\sum_{\beta \in \mathbb{F}}\sum_{f\in \mathbb{F}^{k-1}[x]}\chi((\alpha-\beta) f)\\
&= \sum_{\beta \neq \alpha}\sum_{f\in \mathbb{F}^{k-1}[x]}\chi((\alpha-\beta) f)+ \sum_{f\in \mathbb{F}^{k-1}[x]}\chi(0\cdot f)=Q^{k-1},
\end{align*}
where the last equality follows since for $\alpha\neq \beta$ the inner sum vanishes, as $\chi$ is a nontrivial character. 

For the other direction assume that there is no such $\alpha$ and we need to show that the corresponding eigenvalue is zero.  Then,  for any $\alpha\in \mathbb{F}$ there exists a polynomial $f$ of degree less than $k-1$ such that $\chi(xf(x))\neq \chi(\alpha f(x))$, equivalently, $\chi((x-\alpha)f(x))\neq 1$. This implies that $\chi$ is not constant on polynomials of degree less than $k$ that have a root  at $\alpha$. For $\alpha\in \mathbb{F}$ let  $f_{\alpha}$
be  some  polynomial of degree less than $k-1$ that satisfies 
\begin{equation}
\label{good}
    \chi((x-\alpha)f_\alpha(x))\neq 1,
\end{equation}
whose existence is guaranteed by the preceding discussion. Then for $\alpha\in \mathbb{F}$ 

\begin{align}
\sum_{f\in \mathbb{F}^{k-1}[x]}\chi((x-\alpha)f)&=\sum_{f\in \mathbb{F}^{k-1}[x]}\chi((x-\alpha)(f+f_\alpha))\nonumber\\
&=\chi((x-\alpha)f_\alpha)\cdot\sum_{f\in \mathbb{F}^{k-1}[x]}\chi((x-\alpha)f)=0\label{good2},
\end{align}
where the last equality follows from  \eqref{good}. Then, the eigenvalue equals  
$$
    \overline{<z,\chi>}=\sum_{f\in \mathbb{F}^k[x]}z(f)\chi(f)
    =\sum_{f\in \mathbb{F}^{k}[x]} 
    \sum_{\substack{\alpha \in \mathbb{F}\\f(\alpha)=0}}
    \chi((x-\alpha)f)= \sum_{\alpha \in \mathbb{F}}\sum_{f\in \mathbb{F}^{k-1}[x]}\chi((x-\alpha)f)=0,
$$
where the last equality follows from \eqref{good2}, and the result follows. 
\end{proof}

\Cref{lemma3} provided a necessary and sufficient condition for a nontrivial character to have a nonzero eigenvalue. The next lemma offers a more explicit characterization of those characters, facilitating the enumeration of their quantity.

\begin{lemma}
\label{lemma2}
Let $\alpha \in \mathbb{F}$. The character $\chi_v$ satisfies 
\begin{equation}
\label{good3}
    \chi_v(\alpha \cdot f) = \chi_v(x \cdot f) \text{ for all } f \in \mathbb{F}^{k-1}[x]
\end{equation}
if and only if 
\begin{equation}
\label{eq:1}
    v = \beta \cdot (1, \alpha, \ldots, \alpha^{k-1})
\end{equation}
for some $\beta \in \mathbb{F}$.
\end{lemma}

\begin{proof}
First, for $v = 0$, it is straightforward to verify that both \eqref{good3} and \eqref{eq:1} hold. Next, by \Cref{lemma3}, a character $\chi_v$ satisfying \eqref{good3} corresponds to the eigenvalue $Q^{k-1}$. Conversely, by \Cref{eigenvalues2}, the eigenspace of $T^* \cdot T$ associated with the eigenvalue $Q^{k-1}$ has dimension $Q(Q-1)$. Therefore, the eigenspace of $T \cdot T^*$ corresponding to the eigenvalue $Q^{k-1}$ also has dimension $Q(Q-1)$. Since the characters are linearly independent, we conclude that there are exactly $Q(Q-1)$ characters satisfying \eqref{good3}. We next demonstrate that these are precisely the $Q(Q-1)$ vectors of the form \eqref{eq:1} for $v \neq 0$. Indeed, there are $Q(Q-1)$ such non-zero vectors, as there are $Q-1$ options to select $\beta$ and $Q$ options to select $\alpha$. Lastly, we show that if $v$ satisfies \eqref{eq:1}, then \eqref{good3} holds.

Indeed, for $f\in \mathbb{F}^{k-1}[x]$ the left-hand side of \eqref{good3} becomes
\[ \chi_v(\alpha f) = e(\tr(\beta \cdot \alpha \cdot f(\alpha))) = e(\tr((\beta \cdot x \cdot f)(\alpha))) = \chi_v(x \cdot f), \]
as required.
\end{proof}

By combining \Cref{lemma3} and \Cref{lemma2} we get a decomposition of the space $\mathbb{C}[\mathbb{F}^k[x]]$ to eigenvectors and their eigenvalues.

\begin{thm}
\label{eigenvalues1}
The characters of $\mathbb{F}^k[x]$ are eigenvectors of the operator $\phi_z$ with the following eigenvalues:

\begin{center}
\begin{tabular}{ |c |c| c|  }
\hline
 \text{Character $\chi_v$}& \text{No. of characters} & \text{Eigenvalue}   \\ 

  \hline
 $v=0$ & $1$  & $Q^{k}$  \\  
\hline $v=\beta(1,\alpha,\ldots,\alpha^{k-1}), \beta\neq 0, \alpha\in \mathbb{F}$ & $Q(Q-1)$ & $Q^{k-1}$\\
\hline  else &$Q^k - Q(Q-1)-1$ & $0$\\
  \hline
\end{tabular}
\end{center}
\end{thm}

\begin{proof}
Let $\chi_0$ be the trivial character, then its corresponding eigenvalue equals
$$<z,\chi_0>=\sum_{f\in \mathbb{F}^k[x]}z(f)\overline{\chi_0(f)}=\sum_{\beta\in \mathbb{F}}\sum_{\substack{f\in \mathbb{F}^k[x]\\f(\beta)=0}}1=\sum_{\beta\in \mathbb{F}}Q^{k-1}=Q^k.$$
This describes the first row of the table. 
The remaining two rows of the table follow from Lemma \ref{lemma2} and Lemma \ref{lemma3}. 
\end{proof}
Similar to \Cref{cor-right-singular-vectors} we get the following corollary. 
\begin{cor}
The normalized characters of $\mathbb{F}^k[x]$ serve as left-singular vectors of the matrix $T$, each associated with a singular value equal to the square root of their corresponding eigenvalue, as detailed in Theorem \ref{eigenvalues1}.
\end{cor}


Similar to \Cref{decomposition2}, by building on \Cref{eigenvalues1} we establish a bound on the squared norm of the projection of an indicator vector onto the eigenspaces. This bound, presented in the following lemma, will be instrumental in obtaining the results of  Section \ref{applicatoins}. But first we will need the following notation. For a polynomial $f$ we denote by $Z_\mathbb{F}(f)=|\{\alpha\in \mathbb{F}:f(\alpha)=0\}|$ the number of roots $f$ has in the field $\mathbb{F}.$
\begin{lemma}
\label{tikitaka}
Let $1_L\in \mathbb{C}^{Q^k}$ be the indicator vector of a set $L\subseteq \mathbb{F}^k[x]$ of $\ell$ polynonmials in the plane. Then, the squared norm of the projection of $1_L$ onto  the eigenspaces with eigenvalues $Q^k$ and $Q^{k-1}$ is 
$\ell/Q^{k/2} \text{ and  at most }  Q^{-k+1}\ell (Q+\ell(k-1))$, respectively.  
\end{lemma}
\begin{proof}
The squared norm of the projection  of $1_L$ onto the normalized trivial character is 
$$\left|\left\langle 1_L,\frac{\chi_0}{Q^{k/2}}\right\rangle\right|^2=\frac{\ell^2}{Q^{k}},$$
as needed.

Denote by  $\chi_{\alpha,\beta}$ the character $\chi_v$ for $v=\beta(1,\alpha,\ldots,\alpha^{k-1})$. Next, we bound the squared norm of the projection of $1_L$ onto the eigenspace corresponding to the eigenvalue $Q^{k-1}$, which is calculated as follows:

\begin{align}
   &\sum_{\substack{\alpha,\beta\in \mathbb{F}\nonumber\\ 
    \beta\neq 0}}\big|\big<\frac{\chi_{\alpha,\beta}}{Q^{k/2}},1_L\big>\big|^2=Q^{-k}\sum_{\substack{\alpha,\beta\in \mathbb{F}\\ \beta\neq 0}}\sum_{f\in L}\chi_{\alpha,\beta}(f)\sum_{f'\in L}\overline{\chi_{\alpha,\beta}(f')}= Q^{-k}\sum_{f, f'\in L}\sum_{\substack{\alpha,\beta\in \mathbb{F}\\ \beta\neq 0}}\chi_{\alpha,\beta}(f-f')\nonumber\\
    &= Q^{-k}(\ell Q(Q-1)+\sum_{\substack{f,f'\in L\\f\neq  f'}}\sum_{\substack{\alpha,\beta\in \mathbb{F}\\ \beta\neq 0}}\chi_{\alpha,\beta}(f-f'))\nonumber\\
    &= Q^{-k}(\ell Q(Q-1)+\sum_{\substack{f,f'\in L\\f\neq  f'}}\sum_{\substack{\alpha\in \mathbb{F}\\ f(\alpha)= f'(\alpha)}}\sum_{\beta\neq 0}\chi_{\alpha,\beta}(f-f')+\sum_{\substack{f,f'\in L\\f\neq  f'}}\sum_{\substack{\alpha\in \mathbb{F}\\ f(\alpha)\neq f'(\alpha)}}\sum_{\beta\neq 0}\chi_{\alpha,\beta}(f-f'))\nonumber\\
 &= Q^{-k}(\ell Q(Q-1)+\sum_{\substack{f,f'\in L\\f\neq  f'}}|Z_\mathbb{F}(f-f')|(Q-1)+\sum_{\substack{f,f'\in L\\f\neq  f'}}\sum_{\substack{\alpha,\beta\in \mathbb{F}\\ f(\alpha)\neq f'(\alpha)}}-1)\label{good4}\\
&= Q^{-k}(\ell Q(Q-1)+\sum_{\substack{f,f'\in L\\f\neq  f'}}|Z_\mathbb{F}(f-f')|(Q-1)-\sum_{\substack{f,f'\in L\\f\neq  f'}}(Q-|Z_\mathbb{F}(f-f')|))\nonumber\\
&\leq  Q^{-k}(\ell Q(Q-1)+\ell (\ell-1)(k-1)(Q-1)-\ell(\ell-1)(Q-k+1))\label{good5}\\
& \leq Q^{-k+1}\ell(Q+\ell(k-1)),\nonumber
\end{align}
where \eqref{good4} follows since whenever  $f(\alpha)\neq f'(\alpha)$, $\beta(f-f')(\alpha)$ attains all  nonzero elements of the field when $\beta$ ranges in $\mathbb{F}\backslash \{0\}$. \eqref{good5} follows since $f-f'$ is a nonzero polynomial of degree at most $k-1$. 
\end{proof}

\begin{remark}
\label{remark1}
Note that for $k=2$ one can improve slightly the upper bound in  \Cref{tikitaka} on the squared norm of the projection of $1_L$ onto  the eigenspaces with eigenvalue $Q^{k-1}$
 as \eqref{good5} becomes $\frac{\ell(Q-1)}{Q}$.
\end{remark}

\section{Applications}
\label{applicatoins}
In this section we collect some of the results obtained in the previous sections to derive the bound on the number of point-polynomial incidences. Then, we apply the new bound to the problem of average-radius  list-decoding of RS codes. 

 We begin with the proof of \Cref{main-thm} which is restated for the reader's convenience. 
\subsection{A Bound on Point-Polynomial Incidences}
\label{applications-2}

\mainthm*

\begin{proof}
Let $T$ be the $Q^k \times Q^2$ incidence matrix between polynomials of degree less than $k$ and points in the plane. Let $|L| = \ell$ and $|P| = p$, and express the indicator vector $1_L$ as
\begin{equation}
    \label{pikipiki}
    1_L = \frac{\ell}{Q^{\frac{k}{2}}}v_0 + \beta_1 v_1 + \beta_2 v_2,
\end{equation}
where $v_0, v_1,$ and $v_2$ are unit vectors in the eigenspaces corresponding to eigenvalues $Q^k, Q^{k-1},$ and $0$, respectively. This decomposition and the bound $|\beta_1|^2 \leq Q^{-k+1}\ell(Q + \ell(k-1))$ follow from \Cref{tikitaka}.

Similarly, express the indicator vector $1_P$ as
\[ 1_P = \frac{p}{Q}u_0 + \alpha_1 u_1 + \alpha_2 u_2, \]
where $u_0, u_1,$ and $u_2$ are unit vectors in the eigenspaces corresponding to eigenvalues $Q^k, Q^{k-1},$ and $0$, respectively. This decomposition and the bound $|\alpha_1|^2 \leq \frac{(Q-1)p}{Q}$ follow from \Cref{decomposition2}. By the singular value decomposition of $T$, we have $Tu_0 = Q^{\frac{k}{2}}v_0$ and $Tu_1$ is orthogonal to $v_0, v_2$. 

The number $I(P, L)$ of incidences between points and polynomials satisfies
\begin{equation}
    \label{pikipi2}
    I(P, L) = 1_L^t \cdot T \cdot 1_P = 1_L^tT \left( \frac{p}{Q}u_0 + \alpha_1 u_1 + \alpha_2 u_2 \right) = 1_L^t \cdot \left( Q^{\frac{k}{2}}\frac{p}{Q}v_0 + Q^{\frac{k-1}{2}}\alpha_1 \tilde{v}_1 \right),
\end{equation}
where $Tu_1=Q^{\frac{k-1}{2}}\tilde{v}_1$ and $\tilde{v}_1$ is a unit vector. Substituting \eqref{pikipiki} into \eqref{pikipi2}, we find that $I(L, P)$ equals
\[ \left( \frac{\ell}{Q^{\frac{k}{2}}}v_0 + \beta_1 v_1 + \beta_2 v_2 \right)^t \cdot \left( Q^{\frac{k}{2}}\frac{p}{Q}v_0 + Q^{\frac{k-1}{2}}\alpha_1 \tilde{v}_1 \right) = \frac{\ell p}{Q} + Q^{\frac{k-1}{2}}\alpha_1\beta_1 \langle v_1, \tilde{v}_1 \rangle. \]
Hence,
\begin{align}
    \left| I(P, L) - \frac{|L||P|}{Q} \right| &= \left| Q^{\frac{k-1}{2}}\alpha_1\beta_1 \langle v_1, \tilde{v}_1 \rangle \right| \nonumber \\
    &\leq \sqrt{\ell(Q + \ell(k-1))} \sqrt{p \left(1 - \frac{1}{Q}\right)} |\langle v_1, \tilde{v}_1 \rangle| \label{good7} \\
    &\leq \sqrt{\ell p(Q + \ell(k-1)) \left(1 - \frac{1}{Q}\right)}, \label{good8}
\end{align}
where \eqref{good7} follows from the upper bounds on $|\alpha_1|$ and $|\beta_1|$, and \eqref{good8} follows from the Cauchy-Schwarz inequality and the fact that $v_1, \tilde{v}_1$ are unit vectors.
\end{proof}

Next, we prove \Cref{vinh-improvement}, which is a modest improvement for $k=2$ to the result of \cite[Theorem 3]{VINH20111177} on the number of point-line incidences.  \Cref{vinh-improvement} is restated for convenience. 
\vinhimprove*

\begin{proof}
Following the proof \Cref{main-thm}, by \Cref{remark1} $|\beta_1|\leq \sqrt{\frac{\ell(Q-1)}{Q}}$, therefore \eqref{good7} can be replaced by  $$Q^\frac{1}{2}\sqrt{\frac{\ell(Q-1)}{Q}}\sqrt{\frac{p(Q-1)}{Q}}=
\sqrt{p\ell Q}\left(1-\frac{1}{Q}\right)$$
\end{proof}

\subsection{An Application to Average-Radius List-Decoding of RS Codes}
\label{application-to-list-decoding}
In this section, we apply the main result of Section \ref{applications-2}, namely \Cref{main-thm}, to the problem of average-radius list-decoding of  RS codes. We begin with the necessary definitions to establish the context.

Let $\mathbb{F}$ be a finite field of order $Q$, and let $n$ be a positive integer. A code is a subset $\mathcal{C} \subseteq \mathbb{F}^n$. The rate $R$ of the code $\mathcal{C}$ is defined as
\[ R = \frac{\log_q |\mathcal{C}|}{n}, \]
where the rate lies within the interval $[0, 1]$.

For $p \in (0, 1)$, called the list-decoding radius, and an integer $\ell \geq 1$, a code $\mathcal{C} \subseteq \mathbb{F}^n$ is said to be $(p, \ell)$-list-decodable if, for all $z \in \mathbb{F}^n$,
\begin{equation*}
 |\{ c \in \mathcal{C} : \delta(c, z) \leq  p \}| \leq  \ell, 
\end{equation*}
where $\delta(x, y) = \frac{1}{n} |\{i : x_i \neq y_i\}|$ denotes the relative Hamming distance. Informally, $\mathcal{C}$ is list-decodable if not too many codewords of $\mathcal{C}$ are contained within any sufficiently small Hamming ball.

Average-radius list-decoding is a variation of list-decoding. A code $\mathcal{C}$ is $(p, \ell)$-average-radius list-decodable if for any set $L \subseteq \mathcal{C}$ of size $\ell+1$ and $z \in \mathbb{F}^n$,
\begin{equation}
    \label{average-radius-def}
    \frac{1}{\ell+1} \sum_{c \in L} \delta(c, z) > p. 
\end{equation}
It is evident that $(p, \ell)$-average-radius list-decodability implies $(p, \ell)$-list-decodability. To illustrate, consider the scenario where there are $\ell+1$ codewords $c \in \mathcal{C}$ each having a relative distance at most $p$ from some $z \in \mathbb{F}^n$. In this case, the average distance of these codewords from $z$ would also be at most $p$, leading to a contradiction of \eqref{average-radius-def}. Thus, average-radius list-decodability is indeed a stronger property than list-decodability. However, it is worth noting that list-decodability has been more extensively studied in coding theory compared to average-radius list-decodability.

We now define RS codes. For positive integers $k < n \leq  Q$, an $[n,k]$ RS code $\subseteq \mathbb{F}^n$ over the field $\mathbb{F}$  is defined by $n$ distinct evaluation points $\alpha_1, \ldots, \alpha_n \in \mathbb{F}$, and it consists of the codewords
\[ \{(f(\alpha_1), \ldots, f(\alpha_n)) : f \in \mathbb{F}^k[x]\}. \]
Note that the rate $R$ of the RS code equals $k/n.$

The well-known Guruswami-Sudan list-decoding algorithm \cite{Guru-sudan-algo} is capable of efficiently list-decoding RS codes of rate $R$ up to a radius of $1 - \sqrt{R}$, called the Johnson bound, with a constant list size. In particular, for any $\varepsilon>0$ the RS code is list-decodable from a radius $1-\sqrt{R}-\varepsilon$ and list size that is constant for constants $R$ and $\varepsilon$.
By applying \Cref{main-thm}, we obtain a similar result but for the stronger notion of  average-radius list-decdong of RS codes.

\begin{thm}
\label{list-decoding-cor1}
Let $\epsilon>0$. Then, an RS code over $\mathbb{F}$, of length $n$ and rate $R$  is \mbox{$(1-\sqrt{R}-\epsilon,\ell)$} average-radius list-decodable with $\ell \leq \left\lceil \frac{Q}{2\epsilon n \sqrt{R}} \right\rceil$.
\end{thm}

\begin{proof}
Consider $\ell+1$ codewords $c_1, \ldots, c_{\ell+1}$  of the RS code, and let $L \subseteq \mathbb{F}^k$ be the set of $\ell+1$ polynomials corresponding to these codewords. Let $(z_1, \ldots, z_n) \in \mathbb{F}^n$ and define $P = \{(\alpha_i, z_i) : i = 1, \ldots, n\}$ as a set of $n$ points in the plane $\mathbb{F}^2$. By \Cref{main-thm}, the number of incidences $I(L, P)$ between $L$ and $P$ satisfies
\[ I(L, P) \leq (\ell+1) n \left(\frac{1}{Q} + \sqrt{\frac{Q}{(\ell+1) n} + R}\right) \leq (\ell+1) n \left(o(1) + \sqrt{R}\left(1 + \frac{Q}{2(\ell+1) nR}\right)\right), \]
where the $o(1)$ term tends to zero as the alphabet size $Q$ tends to infinity, and the last inequality follows since $\sqrt{1+x}\leq 1+\frac{x}{2}$ for $x\geq 0$.

This implies that $(\ell+1) n - I(L, P)$, the number of non-incidences, is at least
\[ (\ell+1) n \left(1 - o(1) - \sqrt{R}\left(1 + \frac{Q}{2(\ell+1) nR}\right)\right), \]
which is greater than  $(\ell+1) n(1 - \sqrt{R} - \epsilon)$ for $\ell = \left\lceil \frac{Q}{2\epsilon n \sqrt{R}} \right\rceil$ and sufficiently large field size $Q$. Equivalently,
\[ \frac{1}{\ell+1} \sum_{i=1}^{\ell+1} \delta(c_i, z) > 1 - \sqrt{R} - \epsilon, \]
and the result follows.
\end{proof}
%
\Cref{list-decoding-cor1} suggests that for an RS code with length $n = \Omega(Q)$, the list size $\ell$ remains constant for a fixed $R$ and $\epsilon$. More precisely, $\ell = O((\epsilon \sqrt{R})^{-1})$. Specifically, for a full-length RS code, where the code length equals the size of the field, we obtain the following result.
\begin{cor}
\label{cor-full-length}
 A full-length RS code of rate $R$ is 
 $\left(1-\sqrt{R}-\epsilon,\left\lceil \frac{1}{2\epsilon  \sqrt{R}} \right\rceil\right)$ average radius list-decodable. 
\end{cor}
\begin{remark}
We note that \Cref{list-decoding-cor1}, and consequently \Cref{cor-full-length}, are also applicable to  list-recovery of RS codes, a generalization of list-decoding \cite{rudra,Guruswami-rudra-limits-list-decoding,LP20}. We omit the details.
\end{remark}

\bibliographystyle{alpha}
	\bibliography{bibliography}

\end{document}